\documentclass[10pt,conference]{IEEEtran}
\IEEEoverridecommandlockouts{}
\usepackage{amsfonts}
\usepackage{amssymb}
\usepackage{stfloats}
\usepackage{cite}
\usepackage{graphicx}
\usepackage{psfrag}
\usepackage{amsmath}
\usepackage{array}
\usepackage{epstopdf}
\usepackage{authblk}
\usepackage{graphicx} 
\usepackage{amsthm} 
\usepackage{lipsum}
\usepackage{verbatim} 
\usepackage{color}
\usepackage{caption}
\usepackage{subcaption}
\usepackage{amsmath,amsfonts,amssymb,mathrsfs}
\usepackage{bm}
\usepackage{dsfont}
\usepackage{upgreek}

\newtheorem{Definition}{Definition}
\newtheorem{Lemma}{Lemma}

\DeclareMathOperator*{\argmax}{arg\,max}
\DeclareMathOperator*{\argmin}{arg\,min}

\newcommand*{\vehSet}{\mathcal{N}}
\newcommand*{\waypSet}{\mathcal{W}}
\newcommand*{\vehNum}{N}

\newcommand*{\bsSet}{\mathcal{M}}
\newcommand*{\bsNum}{M}
\newcommand*{\txTime}{\tau}

\newcommand*{\periodLen}{{T_p}}
\newcommand*{\slotLen}{{T_s}}
\newcommand{\bandwidth}{\mathrm{B}}

\newcommand{\p}[1]{\textbf{P#1}}

\newcommand{\constraint}[1]{\textit{C#1}}

\usepackage{geometry}
\def\BibTeX{{\rm B\kern-.05em{\sc i\kern-.025em b}\kern-.08em
    T\kern-.1667em\lower.7ex\hbox{E}\kern-.125emX}}

\begin{document}
\title{A Dynamic Improvement Framework for Vehicular Task Offloading
}

\author{Qianren Li, Yuncong Hong, Bojie Lv, and Rui Wang\thanks{ Rui Wang is the corresponding author. \par This work was supported in part by National Natural Science Foundation of China under Grant 62171213, and in part by High level of special funds under Grant G03034K004. }}
\maketitle
\begin{abstract}
    In this paper, the task offloading from vehicles with random velocities is optimized via a novel dynamic improvement framework. Particularly, in a vehicular network with multiple vehicles and base stations (BSs), computing tasks of vehicles are offloaded via BSs to an edge server. Due to the random velocities, the exact trajectories of vehicles cannot be predicted in advance. Hence, instead of deterministic optimization, the cell association, uplink time and throughput allocation of multiple vehicles in a period of task offloading are formulated as a finite-horizon Markov decision process. In the proposed solution framework, we first obtain a reference scheduling scheme of cell association, uplink time and throughput allocation via deterministic optimization at the very beginning. The reference scheduling scheme is then used to approximate the value functions of the Bellman's equations, and the actual scheduling action is determined in each time slot according to the current system state and approximate value functions. Thus, the intensive computation for value iteration in the conventional solution is eliminated. Moreover, a non-trivial average cost upper bound is provided for the proposed solution framework. In the simulation, the random trajectories of vehicles are generated from a high-fidelity traffic simulator. It is shown that the performance gain of the proposed scheduling framework over the baselines is significant.
\end{abstract}

\section{Introduction}
Edge or cloud computing for vehicular networks has become necessary in many emerging applications. For example, the online ride-hailing services, such as Uber and Didi Chuxing, have gained widespread popularity. 
In order to address the safety risks \cite{UbersUSSafety}, the data of safety monitoring in the vehicle should be offloaded to remote servers. However, limited wireless radio resource for task offloading may hinder real-time response of edge or cloud computing. Moreover, the high mobility of vehicles may bring uncertainty in the uplink transmission scheduling, and hence, degrade the scheduling efficiency.

In fact, there have been a number of research efforts devoted to the efficient scheduling of task offloading in vehicular networks \cite{zhangJointOffloadingDecision2021, jiaLearningBasedQueuingDelayAware2021, choEnergyEfficientCooperativeOffloading2022,malekiHandoverEnabledDynamicComputation2023, zhangMDPBasedTaskOffloading2020, liuDeepReinforcementLearning2019}.
For example, in \cite{zhangJointOffloadingDecision2021}, with the knowledge of the vehicle's velocity, channel allocation and base station association were optimized to minimize the transmission latency of computing tasks. 
Similarly, in \cite{jiaLearningBasedQueuingDelayAware2021}, with prior knowledge of base stations (BSs) association, throughput was maximized by selecting offloading strategies. 
Furthermore, energy consumption for both computation and transmission was minimized in \cite{choEnergyEfficientCooperativeOffloading2022} by optimizing offloading power and task sizes, where constant vehicles' velocities were assumed. In all these works, it was assumed that the trajectories of vehicles are deterministic and known to the scheduler. Hence, the scheduling design of task offloading became deterministic optimization problems. However, in practice, even if the vehicles are moving along predetermined routes, their velocities can be random due to complicated traffic conditions. 

Taking the randomness of trajectories into consideration, the scheduling of task offloading in vehicular networks usually becomes stochastic optimization problems, which could be formulated as Markov decision process (MDP). For instance, a double deep Q-network (DDQN) was employed in \cite{malekiHandoverEnabledDynamicComputation2023} to optimize the offloading and computing decisions for vehicles moving randomly within a city. 
In \cite{zhangMDPBasedTaskOffloading2020}, value iteration was adopted to optimize the BS association, while \cite{liuDeepReinforcementLearning2019} presented a deep reinforcement learning algorithm for optimizing the task offloading strategies. 
Despite their effectiveness, these methods are computationally intensive in model training or value function evaluation. This is because the training or evaluation relies on numerical methods. Thus, no analytical model could be provided for the value function or Q-function. Therefore, their performance could hardly be investigated analytically.

In this paper, we would like to shed some light on the stochastic optimization of task offloading in a vehicular network. Particularly, we consider the scheduling of task offloading from multiple vehicles to an edge server via multiple BSs. Due to the randomness in the vehicles' velocities, we formulate the BS association, uplink time and throughput allocation for task offloading as a finite-horizon MDP. A novel low-complexity solution framework, namely dynamic improvement framework, is proposed, where the value function of the Bellman's equations can be approximated analytically without conventional complicated value iteration. A non-trivial upper bound on the average scheduling cost is also derived for the proposed approximate MDP solution. Finally, a high-fidelity traffic simulator is adopted to generate the trajectories of vehicles, and significant performance gain of the proposed scheduling framework is demonstrated in the simulation.

\section{System Model}
\label{sec: System Model}
The task offloading of $\vehNum$ vehicles in a scheduling period of $\periodLen$ seconds via a cellular network of $\bsNum$ BSs is considered in this paper. The sets of vehicles and BSs are denoted as $\vehSet = \{ 1, 2, \ldots, \vehNum\}$ and $\bsSet = \{1,2,\ldots, \bsNum\}$, respectively. The formers are moving along predetermined routes with random velocities. One computing task is generated at each vehicle in the scheduling period and uploaded to an edge server via the BSs. The edge server is connected with all the $M$ BSs. The tasks are offloaded when the vehicles are moving, there could be handover during task offloading. The BSs operate on different frequency bands, each with a bandwidth of $\bandwidth$ Hz, to avoid inter-cell interference. The vehicles' mobility model, uplink transmission model, and the tasks' queuing model are elaborated below.

\subsection{Mobility Model of Vehicles}
The predetermined routes of all vehicles are known to a centralized offloading scheduler as in the popular ride-hailing service (e.g., Didi and Uber). To facilitate the offloading scheduling, each route is quantized to a sequence of waypoints. For example, the route of the $n$-th vehicle ($\forall n$) is represented as  $\waypSet_n \triangleq (\boldsymbol{l}^0_n, \boldsymbol{l}^1_n, \ldots, \boldsymbol{l}^{|\waypSet_n|-1}_n)$, where  $\boldsymbol{l}^0_n$ is the initial position of the $n$-th vehicle at the beginning of the scheduling period, and $\boldsymbol{l}^i_n \in \mathbb{R}^{2 \times 1}$ is the position of the $i$-th waypoint. 

The position of the $n$-th vehicle versus time in the considered scheduling period, referred to as its trajectory, can be modeled as a stochastic process on $\waypSet_n$. Particularly, the scheduling period $\periodLen$ is divided into $T$ time slots, and the slot duration is $\slotLen=\periodLen/T$. Let $\boldsymbol{l}_{n,t} \in \waypSet_n$ represents the location of $n$-th vehicle in the $t$-th slot. The trajectory of the $n$-th vehicle, denoted as $(\boldsymbol{l}_{n,1}, \boldsymbol{l}_{n,2}, \ldots, \boldsymbol{l}_{n,T})$, is modeled as a stationary Markov chain, where the transition probability matrix $\mathbf{P}_n \in \mathbb{R}^{\vert \waypSet_n \vert \times \vert \waypSet_n \vert}$ is defined as
\begin{equation}
 [\mathbf{P}_n]_{ij} \triangleq \mathds{P}\left[\boldsymbol{l}_{n,t+1} = \boldsymbol{l}^{i}_n | \boldsymbol{l}_{n, t} = \boldsymbol{l}^j_n\right], \forall i, j,
\end{equation}
and $[\mathbf{P}_n]_{ij}$ is the $(i,j)$-th element of matrix $\mathbf{P}_n$. Moreover, the aggregation matrix $\boldsymbol{L}_t \in \mathbb{R}^{2 \times \vehNum}$ of $\vehNum$ vehicles' locations in $t$-th slot is given by
\begin{equation}
 \boldsymbol{L}_t \triangleq [\boldsymbol{l}_{1,t}, \boldsymbol{l}_{2,t}, \ldots, \boldsymbol{l}_{\vehNum,t}].
\end{equation}
It is assumed that transition probability matrix $\mathbf{P}_n $ is known to the centralized offloading scheduler. In fact, a high-fidelity simulator has been proposed in \cite{zhang2021distributed} to obtain the above matrix for arbitrary road maps and traffic scenarios. 

\subsection{Uplink Model}
\label{section: Channel Model}

Let $\{\boldsymbol{l}_1,...,\boldsymbol{l}_\bsNum\}$ be the locations of $\bsNum$ BSs, respectively. In each time slot, one vehicle can be associated with at most one BS for the uplink transmission of its computing task. Let $e_{n,m,t} \in \mathcal{E} \triangleq \{0, 1\}$ be the indicator of BS association, where $e_{n,m,t} = 1$ indicates the association of the $n$-th vehicle with the $m$-th BS in the $t$-th time slot and $e_{n,m,t} = 0$ otherwise.
Define the association matrix of the $t$-th time slot $\boldsymbol{E}_t \in \mathcal{E}^{\vehNum \times \bsNum}$ as
$
 \left[\boldsymbol{E}_t\right]_{n,m} = e_{n,m,t}.
$
The constraint of unique BS association can be written as
\begin{equation}
 \constraint{1}: \boldsymbol{E}_t \mathbf{1} = \mathbf{1}, \forall t,
\end{equation}
where $\mathbf{1}$ is the column vector of 1s.

The time-division multiple access (TDMA) mechanism is adopted in each cell. Let $\tau_{n,t}$ be the uplink time ratio of the associated cell allocated to the $n$-th vehicle in the $t$-th time slot, and $ \boldsymbol{\txTime}_t \triangleq [\txTime_{1,t}, \txTime_{2,t}, \ldots, \txTime_{\vehNum,t}]^T \in \mathbb{R}^{\vehNum\times 1} $ be the aggregation vector of time ratios. The uplink transmission time of the $n$-th vehicle in the $t$-th time slot is $\tau_{n,t} T_s$, and the transmission time constraints of all $M$ cells can be written as
\begin{equation}
 \constraint{2}: \mathbf{0} \preceq \boldsymbol{\txTime}_t, \forall t, \quad
 \constraint{3}:  \boldsymbol{\txTime}^T \boldsymbol{E}_t \preceq \mathbf{1}, \forall t,
\end{equation}
where $\preceq$ denotes the element-wise greater or equal relationship, and $\mathbf{0} $ is the column vector of 0s.

In the $t$-th time slot, given the locations of $n$-th vehicle, the path loss of its uplink channel to the $m$-th BS is modeled as 
\begin{equation}
 g_{n,m,t} = G \frac{1}{\| \boldsymbol{l}_{n,t} - \boldsymbol{l}_m \|_2^\gamma},
\end{equation}
where $G$ is a constant, and $\gamma$ denotes the pathloss exponent. It is assumed that the time slot duration $T_s$ is sufficiently large, such that the ergodic channel capacity can be achieved. Hence, the throughput of the $n$-th vehicle in the $t$-th time slot (if it is associated with the $m$-th BS) can be written as 
\begin{equation}
    \label{equation: Throughput}
 r_{n,t} = \txTime_{n,t} \slotLen \bandwidth \mathds{E}_h \left[\log_2 \left(1 + \frac{P_{n, t} g_{n,m,t} \vert h \vert^2}{\sigma^2}\right) \right],
\end{equation}
where $\sigma^2$ is the average noise power, the expectation is taken over the fast fading coefficient $h$, and $P_{n,t}$ is the uplink power with the peak power constraint $P_{\max}$. Thus,
\begin{equation}
 \constraint{4}: 0 \leq P_{n,t} \leq P_{\max}, \forall n, t.
\end{equation} 

\subsection{Task Queuing Model}

In practice, the task arrival of the vehicles is not synchronous. It is assumed that the computing task of the $n$-th vehicle arrives at the beginning of the $T_{n,a}$-th time slot, with a size of $D_{n,a}$ information bits. Hence, the vehicle starts the uplink transmission since the $T_{n,a}$-th time slot. The arrival tasks are buffered at the transmitters. Let $d_{n,t}$ be the number of information bits buffered at the $n$-th vehicle at the beginning of the $t$-th time slot. Its uplink queue dynamics are given by
\begin{equation}
 d_{n, t+1} = \begin{cases}
      0,                                   & t < T_n^a, \\
      \left(D_n^a - r_{n,t}\right)^+,      & t = T_n^a, \\
      \left(d_{n, t} - r_{n,t}\right)^+,   & \text{otherwise},
    \end{cases}
\end{equation}
where $\left( \cdot \right)^+=\max \{0, \cdot \}$. 

As a remark, notice that we focus on the offloading scheduling in a scheduling period of $T_p$ seconds. In practice, the time horizon can be divided into a sequence of scheduling periods, and our design can be applied on each scheduling period respectively. Different scheduling periods might be with different sets of vehicles in task offloading. Due to the late arrival, some tasks might not be completely offloaded in one scheduling period. They can be treated as new tasks in the next scheduling period. 

\section{Problem Formulation}
\label{sec: problem formulation}
The joint scheduling of BS association, time allocation and uplink power adaptation (or equivalently throughput adaptation) of all vehicles in all $T$ time slots is considered in this paper. Due to the randomness in the vehicles' mobility, there is uncertainty on the uplink throughput at the start of the scheduling period. Hence, the joint scheduling design is a stochastic optimization problem, which is formulated as a finite-horizon MDP in this section. Particularly, the system state, action, and policy are first defined in the following.
\begin{Definition}[System State]
  \label{def:state}
 The system state in the $t$-th time slot ($\forall t=1,2,...,T$) is a tuple of buffered information bits and the locations of all the vehicles, i.e.,
  \begin{equation}
 \mathsf{S}_t = (\boldsymbol{d}_t, \boldsymbol{L}_t),
  \end{equation}
 where $\boldsymbol{d}_t \triangleq [d_{1,t}, d_{2,t}, \ldots, d_{N,t}]$ is the vector of buffered information bits.
\end{Definition}

\begin{Definition}[Policy]
  \label{def:policy}
 The action in the $t$-th time slot ($\forall t=1,2,...,T$), denoted as $\mathsf{A}_t $, consists of the BS association, uplink time ratio, and uplink throughput\footnote{We define the uplink throughput, instead of uplink power, as the scheduling action for the elaboration convenience. In fact, they are equivalent.} of each vehicle, thus,
  \begin{equation}
 \mathsf{A}_t = (\boldsymbol{E}_t, \boldsymbol{\txTime}_t, \boldsymbol{r}_t),\nonumber
  \end{equation}
 where $\boldsymbol{r}_t\triangleq [r_{1,t}, r_{2,t},..., r_{N,t}]$ is the vector of throughput allocation for all the vehicles. 
 The policy $\Omega_t$ of the $t$-th time slot ($\forall t=1,2,...,T$) is a mapping from the system state to the action, $\Omega_t: \mathsf{S}_t \mapsto \mathsf{A}_t$. The overall policy is defined as the aggregation of policies of all the time slots, i.e.,
  \begin{equation}
    \Omega = ( \Omega_1, \Omega_2, \ldots, \Omega_T ).\nonumber
  \end{equation}
\end{Definition}

The joint scheduling design of this paper is to minimize the offloading latency of all tasks, while saving the uplink energy consumption of all vehicles. Hence, we consider a weighted summation of non-empty buffer penalty and uplink energy consumption as the system cost of each time slot. Moreover, it is possible that the task offloading of one vehicle may not be accomplished at the end of the scheduling period, especially with a late task arrival or poor uplink channel condition. An additional penalty against unaccomplished offloading is considered in the $T$-th (last) time slot. Therefore, the cost function for $n$-th vehicle in the $t$-th time slot is written as
\begin{equation*}
  \label{equation: per-slot cost function}
 \mathsf{c}_{n,t}(\mathsf{S}_{t},\!\mathsf{A}_{t})\!\triangleq\!\begin{cases}
 \mathds{1}( d_{n,t}\!>\!0 )\!+\!\omega_1\!P_{n,t}\txTime_{n,t}\!& t < T\!, \\
 \mathds{1}( d_{n,t}\!>\!0 )\!+\!\omega_1\!P_{n,t}\txTime_{n,t}\!+\!\omega_2d_{n,T+1} & t = T,
  \end{cases}
\end{equation*}
where the indicator function $\mathds{1}(\mathcal{C})$ is $1$ when the event $\mathcal{C}$ is true and $0$ otherwise, $\omega_1$ and $\omega_2$ are two weights. 
The system cost of $t$-th time slot is then given by $\mathsf{c}_{t} (\mathsf{S}_{t}, \mathsf{A}_{t}) = \sum_{n} \mathsf{c}_{n,t} (\mathsf{S}_{t}, \mathsf{A}_{t}) $.
Hence, the average total cost of the system from the $t$-th time slot can be written as
\begin{equation}
  \label{equation:objective}
 \mathsf{C}_t( \mathsf{S}_t, \Omega ) \triangleq \mathds{E}_{\mathsf{S}_{t+1}, \ldots, \mathsf{S}_T} \left[ \sum_{k = t}^{T} \mathsf{c}_{k} (\mathsf{S}_{k}, \mathsf{A}_{k}) \middle| \mathsf{S}_t , \Omega\right], 
\end{equation}
and the average total cost in one scheduling period can be written as $\mathsf{C}_1( \mathsf{S}_1, \Omega )$.
As a result, the joint scheduling problem can be formulated as the following finite-horizon MDP:
\begin{equation*}
  \begin{aligned}
    \p{1}: \min_{\Omega }  \quad  & \mathsf{C}_1( \mathsf{S}_1, \Omega ), \\
    \text{s.t.} \quad & \constraint{1}, \constraint{2}, \constraint{3}, \constraint{4}.
  \end{aligned}
\end{equation*}

Finding the optimal solution of \p{1} is equal to solving the following Bellman's equation~\cite{bertsekas1996dynamic} in each time slot:
\begin{equation}
  \label{equation: Bellman}
 V_{t}(\mathsf{S}_t)\!=\!\min_{\mathsf{A}_t}\!\{\!\mathsf{c}_t(\mathsf{S}_t, \mathsf{A}_t) + \sum_{\mathsf{S}_{t+1}} \mathds{P}\left[ \mathsf{S}_{t+1} | \mathsf{S}_{t}, \mathsf{A}_t \right] V_{t+1}(\mathsf{S}_{t+1})\!\},
\end{equation}
where the optimal value function $V_t(\mathsf{S}_t)$ is the minimum average cost since the $t$-th time slot with the system state $\mathsf{S}_t$. Thus,
\begin{equation}
 V_t(\mathsf{S}_t) = \min_{\Omega_t,\ldots,\Omega_T }\mathds{E}_{\mathsf{S}_{t+1}, \ldots, \mathsf{S}_T} \left[ \sum_{k=t}^{T} \mathsf{c}_k (\mathsf{S}_k, \mathsf{A}_k) \middle| \mathsf{S}_t \right], \forall t.
\end{equation}
With the above optimal value function, the optimal action of the $t$-th time slot ($\forall t=1,2,...,T$) can be obtained by solving the right-hand side (RHS) of the above Bellman's equations as
\begin{equation}\label{equation: Bellman-RHS}
\Omega_t^{\ast}(\mathsf{S}_t)=\argmin_{\mathsf{A}_t}  \mathsf{c}_t(\mathsf{S}_t, \mathsf{A}_t) + \sum_{\mathsf{S}_{t+1}} \mathds{P}\left[ \mathsf{S}_{t+1} | \mathsf{S}_{t}, \mathsf{A}_t \right] V_{t+1}(\mathsf{S}_{t+1}).
\end{equation}

\section{Dynamic Improvement Solution Framework}
\label{sec: Proposed Framework}
The optimal solution of finite-horizon MDP usually suffers from the curse of dimensionality, where the computation complexity increases exponentially with the number of vehicles. 
To address this issue, we propose a low-complexity dynamic improvement solution framework for the problem \p{1}. 

The proposed framework consists of two phases, namely {\it offline deterministic pre-allocation} and {\it online dynamic improvement}. 
In the former phase, each vehicle is assumed to move with a deterministic average trajectory. By removing the randomness in trajectories, the offloading scheduling becomes deterministic. The obtained BS association, time and throughput allocation are referred to as the {\it reference scheduling}. 

In the online phase, the reference scheduling is used to approximate the optimal value functions, namely the {\it approximate value functions}. Then, the scheduling action of each time slot can be optimized according to (\ref{equation: Bellman-RHS}) with the approximate value functions. Moreover, the reference scheduling is also updated after the online optimization of each time slot.

Despite the above approximation, it can be proved that the proposed dynamic improvement framework can successively suppress the average total cost of the reference scheduling according to the current system state. In other words, the performance is lower-bounded by the reference scheduling. The offline pre-allocation, online dynamic improvement and reference scheduling update are elaborated below, respectively.

\subsection{Offline Deterministic Pre-Allocation}
\label{section: offline pre-allocation}
The phase of offline pre-allocation is to provide reference scheduling actions, denoted as $\Omega_0^H \triangleq \{\mathsf{A}_t^{(0)}=(\boldsymbol{E}_t^{(0)}, \boldsymbol{\tau}_t^{(0)}, \boldsymbol{r}_t^{(0)})| \forall t\}$, for all vehicles in all time slots, such that the optimal value function in \eqref{equation: Bellman} can be approximately expressed with a  non-trivial scheduling scheme.

First, we define the following average trajectories:
\begin{equation*}
    \boldsymbol{l}_{n,t}^{(0)} \triangleq \mathds{E}\left[ \boldsymbol{l}_{n,t} | \boldsymbol{l}_{n,1}\right], \forall n, t.
\end{equation*}
Note that above average trajectories are deterministic, the offloading scheduling based on them can be determined at the beginning of scheduling period. Particularly, the BS associations $\{\boldsymbol{E}_t^{(0)}| \forall t\}$ and time allocations $\{ \boldsymbol{\tau}_t^{(0)}| \forall t\}$ of the reference scheduling $\Omega_0^H$ are given by
\begin{equation*}
	[\boldsymbol{E}_t^{(0)}]_{n,m} \triangleq e_{n, m, t}^{(0)} = \mathds{1}(m = \argmax_{k \in \mathcal{\bsNum}} g_{n,k,t}), \forall n,m,t,
\end{equation*}
and
\begin{equation*}
	[\boldsymbol{\txTime}_t^{(0)}]_{n} \triangleq \txTime_{n,t}^{(0)} = 1/ [\mathbf{1}^T \boldsymbol{E}_t^{(0)} ]_m , \forall n,t,
\end{equation*}
where $[\boldsymbol{\txTime}_t^{(0)}]_{n}$ and $[\mathbf{1}^T \boldsymbol{E}_t^{(0)} ]_m$ denote the $n$-th and $m$-th elements of both vectors respectively, and it is assumed that the $n$-th vehicle is associated with the $m$-th BS.

Then, we rely on the deterministic optimization to determine the uplink throughput $\{ \boldsymbol{r}_t^{(0)}| \forall t\}$ of the reference scheduling. Let $\mathcal{L}_{n,t} \triangleq \{ \boldsymbol{l}_{n,t} | \mathds{P}[\boldsymbol{l}_{n,t} | \boldsymbol{l}_{n,1}] > 0 \}$ be the set of possible positions of the $n$-th vehicle in the $t$-th time slot.
In order to ensure the throughput allocation of the reference scheduling is achievable, the following constraints are imposed: 
\begin{equation}
	\constraint{5}: r_{n,t}^{(0)} \le \bar{r}_{n,t}(\tau_{n,t}^{(0)}), \forall n, t,
\end{equation}
where $\bar{r}_{\raisebox{1.5ex}{\scalebox{0.7}{$n,t$}}}$ is the maximum throughput of the $n$-th vehicle in the $t$-th time slot at the worst position, given by
\begin{equation} 
	\bar{r}_{n,t}(\tau_{n,t}^{(0)})  \triangleq \min_{\boldsymbol{l}_{n,t} \in \mathcal{L}_{n,t} }  \txTime_{n,t}^{(0)} \slotLen \mathds{E}_h [\log_2 (1 + \frac{P_{\max} G \vert h \vert^2}{ \Vert \boldsymbol{l}_{n,t} - \boldsymbol{l}_{m} \Vert _2^{\gamma}\sigma^2})]. \nonumber
\end{equation}
Assuming all vehicles are moving along the average trajectories, the expectation on the random trajectories in \eqref{equation:objective} can be omitted. Hence, the stochastic optimization of throughput allocation in \p{1} can be simplified as the following deterministic optimization problem (with additional constraints $\constraint{5}$).
\begin{equation*}
	\begin{aligned}
		\p{2}: \min_{ \{ \boldsymbol{r}_1, \ldots, \boldsymbol{r}_T \}}  \quad  & \sum_{t=1}^{T} \sum_{n \in \mathcal{N}} \mathsf{c}_{n,t} (d_{n,t}, \boldsymbol{l}_{n,t}^{(0)}, [\boldsymbol{E}_t^{(0)}]_n, \txTime_{n,t}^{(0)}, r_{ n,t}), \\
		\text{s.t.} \quad &  \constraint{1}, \constraint{2}, \constraint{3}, \constraint{4}, \constraint{5}.
	\end{aligned}
\end{equation*}
where $[\boldsymbol{E}_t^{(0)}]_n$ denotes the $n$-th row vector of $\boldsymbol{E}_t^{(0)}$.

Given BS associations and time allocations, the throughput allocation of each vehicle in \p{2} can be decoupled. The optimization of throughput allocation of the $n$-th vehicle in all the time slots can be written as
\begin{equation*}
	\begin{aligned}
		\p{3}: \min_{\{ r_{n,1}, \ldots, r_{n,T} \} }  \quad  & \sum_{t=1}^{T} \mathsf{c}_{n,t} (d_{n,t}, \boldsymbol{l}_{n,t}^{(0)}, [\boldsymbol{E}_t^{(0)}]_n, \txTime_{n,t}^{(0)}, r_{ n,t}), \\
		\text{s.t.} \quad & \constraint{1}, \constraint{2}, \constraint{3}, \constraint{4}, \constraint{5}.
	\end{aligned}
\end{equation*}

The optimization of problem \p{3} can be conducted in two cases: the task offloading can be accomplished in the scheduling period and otherwise. Both cases are discussed below.

\textit{1) Case 1}: Suppose the task offloading of the $n$-th vehicle is accomplished in the $t_n^d$-th time slot and  $t_n^d \leq T$, the problem P3 can be rewritten as the following problem \p{4}, where the high SNR approximation is applied on the throughput of \eqref{equation: Throughput}. 
\begin{equation*}
	\begin{aligned}
		\p{4}: \min_{ \left\{ r_{n,t} \middle | t \le t_n^d \right\} } \quad & \omega_1 \sum_{t \le t_n^d} 2^{\frac{r_{n,t}}{\txTime_{n, t}^{(0)}\slotLen \bandwidth }} \Phi_{n,t}  \txTime_{n, t}^{(0)}, \\
		\text{s.t.}  \quad
		& 0 \le r_{n,t} \le \txTime_{n, t}^{(0)}\slotLen \bandwidth \log_2 \frac{P_{\max}}{\Phi_{n,t}}, \forall t, \\
		& r_{n,t} \le \bar{r}_{n,t}, \forall t, \\
		& D_n^a \le \sum_{t \le t_n^d} r_{n,t}, \\
	\end{aligned}
\end{equation*}
where $\Phi_{n,t} \triangleq 2^{-\mathds{E}_h [ \log_2 \frac{g_{n,m,t}\vert h \vert^2}{\sigma^2} ]}$.
Given $t_n^d$, the optimization of the throughput allocation is convex, which can be solved by existing methods. Notice that $t_n^d$ is an integer between $T_n^a$ and $T$, the optimization of $t_n^d$ is conducted by linear search.

\textit{2) Case 2}: The task offloading cannot be accomplished, the optimization problem can be written as
\begin{equation*}
	\begin{aligned}
		\p{5}:\!\min_{ \left\{ r_{n,t} \middle | t \le T \right\} }\!\!\quad & \omega_1 \sum_{t \le T} 2^{\frac{r_{n,t}}{\txTime_{n, t}^{(0)}\slotLen \bandwidth}} \Phi_{n,t}  \txTime_{n, t}^{(0)}  \!+\!\omega_2 (D_n^a - \sum_{t \le T} r_{n,t}), \\
		\text{s.t.}  \quad
		 & 0 \le r_{n,t} \le \txTime_{n, t}^{(0)}\slotLen \bandwidth \log_2 \frac{P_{\max}}{\Phi_{n,t}}, \forall t, \\
		 & r_{n,t} \le \bar{r}_{n,t}(\tau_{n,t}^{(0)}), \forall t, \\
		 & \sum_{t \le T} r_{n,t} < D_n^a.  \\
	\end{aligned}
\end{equation*}
The above problem \p{5} is convex which can be solved with existing methods. After solving problem \p{4} and \p{5}, the optimal throughput allocation of \p{3} is determined by comparing objective values of \p{4} and \p{5}, which is denoted as $\{\boldsymbol{r}_t^{(0)}| \forall t\}$.

\subsection{Online Dynamic Improvement} \label{subsection:dynamic-improvement}
\begin{table*}[!t]
    \centering
    \caption*{TABLE I:\ Definition of $f_{n, t+1}^1(d_{n,t+1})$ and $f_{n, t + 1}^2(d_{n,t+1}, \boldsymbol{l}_{n,t+1})$.}
    \label{table: case function}
        \begin{tabular}{c|c|c}
        \hline
         Condition  & $f_{n, t+1}^1(d_{n,t+1})$ & $f_{n, t + 1}^2(d_{n,t+1}, \boldsymbol{l}_{n,t+1})$  \\ \hline
         $d_{n,t+1} \le 0$ &  $0$      &$0$              \\ \hline
        $ \sum_{\kappa= 1}^{k} r_{n,t+k}^{s} < d_{n,t+1} \le \sum_{\kappa= 1}^{k+1} r_{n,t+k}^{s}$  & $k+1$     &        $\begin{aligned}
            & \sum_{\kappa= 1}^{k }\txTime_{n,t+\kappa}^{(0)} \varPhi_{n, t+\kappa|t+1} 2^{\frac{r_{n,t+\kappa}^{s}}{\txTime_{n,t+\kappa}^{(0)} \slotLen \bandwidth}}                                   \\
            & + \txTime_{n,t+k+1}^{(0)} \varPhi_{n, t+k+1|t+1} 2^{\frac{d_{n, t} - \sum_{\kappa= 0}^{k} r_{n,t+k}^{s}}{ \txTime_{n,t+k+1}^{(0)} \slotLen \bandwidth }}
            \end{aligned}$            \\ \hline
         $\sum_{\kappa= 1}^{T - t} { r_{n,t+\kappa}^{s} } < d_{n,t+1}$  & $T - t + \omega_2\left(d_{n,t+1} - \sum_{\kappa= 1}^{T - t} {r_{n,t + \kappa}^{s} }\right)$  & $ \sum_{\kappa= 1}^{T-t}\txTime_{n,t+\kappa}^{(0)} \varPhi_{n, t+\kappa|t+1} 2^{\frac{r_{n,t+k}^{s}}{\txTime_{n,t+\kappa}^{(0)}\slotLen \bandwidth}}$     \\ \hline
        \end{tabular}
\end{table*}

The reference scheduling obtained in the previous part will be updated in each time slot, which will be elaborated in Section \ref{subsection:update}. Let $\Omega_t^H\!\!\triangleq\!\! \{\mathsf{A}_k^{(t)}\!\!=\!(\boldsymbol{E}_k^{(0)}, \boldsymbol{\tau}_k^{(0)}, \boldsymbol{r}_k^{(t)})| \forall k\}$ be the reference scheduling updated after the $t$-th time slot. In this part, we consider the online scheduling optimization in the $t$-th time slot, $t\!=\!1,2,...,T$, based on the reference scheduling $\Omega_{t-1}^H$.

Particularly, the scheduling action of the $t$-th time slot can be obtained via \eqref{equation: Bellman-RHS} according to the system state $\mathsf{S}_{t}$. Based on the reference scheduling $\Omega_{t-1}^H$, we first introduce the following conclusion on the achievable average system cost.

\begin{Lemma}[Achievable Average System Cost]
    Given the residual information bits  $d_{n,t+1}$ and position $\boldsymbol{l}_{n,t+1}$ of the $n$-th vehicle at the beginning of the $(t+1)$-th time slot, following the BS association and time allocation of reference scheduling $\Omega_{t-1}^H$, the following average cost of the $n$-th vehicle from the $(t+1)$-th time slot to the last one, denoted as $\widetilde{V}_{n,t+1}$, is achievable by adjusting throughput allocation of $\Omega_{t-1}^H$,
    \begin{multline}
        \widetilde{V}_{n,t+1}( d_{n,t+1},\boldsymbol{l}_{n,t+1}) = f_{n, t+1}^1(d_{n,t+1})  
        \\ + \omega_1 f_{n, t+1}^2(d_{n,t+1}, \boldsymbol{l}_{n,t+1}),
    \end{multline}
    where $ f_{n, t+1}^1(d_{n,t+1}) $ and $f_{n, t+1}^2(d_{n,t+1}, \boldsymbol{l}_{n,t+1})$ are defined in Table {\rm I}, $k=1,2,\ldots,T-t-1$, and the $\varPhi_{n, t+k|t+1}$ is defined as
    \begin{equation*}
        \varPhi_{n, t+k|t+1} \triangleq 2^{-  \mathds{E}[\log_2 \frac{G\vert h \vert^2}{\sigma^2}]} \mathds{E} [\Vert \boldsymbol{l}_{n,t+k} - \boldsymbol{l}_m \Vert_2^{\gamma}| \boldsymbol{l}_{n,t+1} ].
    \end{equation*}
\end{Lemma}
\begin{proof}
    The average system cost can be achievable as follows. 
    Let $\zeta^{(t-1)}$ be the offloading completion time slot of the reference scheduling $\Omega_{t-1}^H$.
    When $\Delta d_{n,t+1} \triangleq d_{n,t+1} - \sum_{k = 1}^{T-t} r_{n,t + k}^{(t-1)} > 0$, the additional $\Delta d_{n,t+1}$ information bits are transmitted from the $\zeta^{(t-1)}$-th time slot with maximum throughput of $ \{r_{n,\zeta^{(t-1)}}^{\mathrm{s}}, r_{n,\zeta^{(t-1)}+1}^{\mathrm{s}}, \ldots,r_{n,T}^{\mathrm{s}} \}$, where maximum throughput $r_{n,t}^{\mathrm{s}}$ is defined in \eqref{equation: auxiliary throughput}. On the other hand, when $d_{n,t+1} < \sum_{k = 1}^{T-t} r_{n,t + k}^{(t)}$, throughput allocation in the last time slot of the reference scheduling $\Omega_{t-1}^H$ can be suppressed.
\end{proof}

Hence, we define the approximate value function as
\begin{equation}
    \widetilde{V}_{t+1}(\mathsf{S}_{t+1}) \triangleq \sum_{n \in \mathcal{N}}\widetilde{V}_{n,t+1}( d_{n,t+1},\boldsymbol{l}_{n,t+1}).
\end{equation}
Replacing the optimal value function of \eqref{equation: Bellman-RHS} by the above approximate value function, the scheduling action optimization of the $t$-th time slot becomes
\begin{equation*}
    \begin{aligned}
    \p{6}: \min_{\mathsf{A}_t} \quad & 
    \mathsf{c}_t(\mathsf{S}_t, \mathsf{A}_t)
    + \sum_{\mathsf{S}_{t+1}} \mathds{P}\left[ \mathsf{S}_{t+1} | \mathsf{S}_{t}, \mathsf{A}_t \right] \widetilde{V}_{t+1}(\mathsf{S}_{t+1}), \\
    \text{s.t.}  \quad & \constraint{1}, \constraint{2}, \constraint{3}, \constraint{4}. \\
    \end{aligned}
\end{equation*}

The above problem \p{6} is convex when $t=T$, which represents the last time slot of scheduling period. 
However, when $t<T$, problem \p{6} becomes non-convex due to the approximate value function. In the following, the alternating minimization method is applied. Particularly, in each iteration, the BS association $\boldsymbol{E}_t$ and time allocation $\boldsymbol{\txTime}_t$ are jointly optimized while keeping the throughput allocation $\boldsymbol{r}_t$ fixed, and then the throughput allocation is optimized with the previously optimized BS association and time allocation.

\subsubsection{Optimization of the BS association and time allocation} 
Note that the optimal solution of BS association should search all the possible combinations, leading to prohibitive complexity. In order to search a sub-optimal association with low complexity, we first relax the problem by allowing each vehicle to associate with all the BSs. Thus, let $\tau_{n,m,t}$ be the uplink time ratio of the $m$-th BS allocated to the $n$-th vehicle in the $t$-th time slot, and  $\eta_{n,m,t} r_{n,t}$ be the information bits of the $n$-th vehicle uploaded via the $m$-th BS in the $t$-th time slot, and $ \mathcal{D}_t = \left( \txTime_{1,1,t}, \txTime_{1,2,t}, \ldots, \txTime_{1,M,t}, \txTime_{2,1,t}, \ldots, \txTime_{N,M,t} \right) $, and $\boldsymbol{\eta}_t = \left( \eta_{1,1,t}, \eta_{1,2,t}, \ldots, \eta_{1,M,t}, \eta_{2,1,t}, \ldots, \eta_{N,M,t} \right)$. Given the throughput allocation, the optimization of the BS association and time allocation with relaxation can be written as
\begin{equation*}
    \begin{aligned}
    \p{7}: \min_{\mathcal{D}_t, \boldsymbol{\eta}_t} \quad
                           & \omega_1 \sum_{n \in \mathcal{N}} \sum_{m \in \mathcal{M}} 2^{\frac{\eta_{n, m,t} r_{n,t}}{\txTime_{n,m,t}\slotLen \bandwidth}} \Phi_{n,t} \txTime_{n,m,t}, \\
    \text{s.t.}  \quad &  0 \le \txTime_{n, m, t}, 0 \le \eta_{n, m,t},\forall n, m, \\
            & \sum_{n \in \mathcal{N}} \txTime_{n,m,t} \le 1, \forall m, \\
            & \sum_{m \in \mathcal{M}} \eta_{n,m,t} = 1, \forall n. \\
    \end{aligned}
\end{equation*}
Note that the above problem \p{7} is convex. Let $\{\eta_{n,m,t}^{\ast}|\forall m\}$ be the optimal solution of $\{\eta_{n,m,t}|\forall m\}$, $\forall n$. The BS association can be determined as 
\begin{equation}
    e_{n,m,t} = \mathds{1} ( m = \argmax_{k \in \mathcal{M}} \eta_{n,k,t}^\ast ).
\end{equation}
Then, the optimal time allocation with the above BS association can be obtained by solving the convex problem \p{7} with $\eta_{n,m,t}=e_{n,m,t}$, $\forall n,m$. 

\subsubsection{Optimization of the throughput}
Given the BS association and time allocation, the throughput optimization of $N$ vehicles can be decoupled, where the $n$-th sub-problem for the $n$-th vehicle can be written as
\begin{equation*}
    \begin{aligned}
    \p{8}: \min_{r_{n,t}} \quad & 
    \omega_1 2^{\frac{r_{n,t}}{\txTime_{n,t}\slotLen \bandwidth}}\Phi_{n,t}\txTime_{n,t} + f_{n, t+1}^1(d_{n,t} - r_{n,t}) \\ 
    &\!+\!\omega_1\!\sum_{ \boldsymbol{l}_{n,t+1} }\!\mathds{P}\left[ \boldsymbol{l}_{n,t+1}|  \boldsymbol{l}_{n,t} \right] f_{n,\!t+1}^2(d_{n,t}\!-\!r_{n,t},\!\boldsymbol{l}_{n,t+1}), \\
    \text{s.t.}  \quad
    & 0 \le r_{n,t} \le \txTime_{n, t}\slotLen \bandwidth \log_2 \frac{P_{\max}}{\Phi_{n,t}}.
    \end{aligned}
\end{equation*}

Due to the structures of $f_{n, t+1}^1(d_{n,t} - r_{n,t})$ and $f_{n, t+1}^2(d_{n,t} - r_{n,t}, \boldsymbol{l}_{n,t+1})$, the objective of problem \p{8} is convex in $T-t+1$ regions of $r_{n,t}$, respectively. Hence, the optimal solution of $r_{n,t}$ is obtained by comparing the values of the objective.

The optimization of problem \p{7} and \p{8} can be applied iteratively until convergence. Denote the optimized BS association, time and throughput allocation after convergence as $\mathsf{A}_t^{\ast}=(\boldsymbol{E}_t^{\ast}, \boldsymbol{\txTime}_t^{\ast}, \boldsymbol{r}_t^{\ast})$.
In order to ensure the performance improvement, it is necessary to check if the optimized actions lead to lower average cost, thus,
\begin{eqnarray}
    &&\mathsf{c}_t\left(\mathsf{S}_t, \mathsf{A}_t^{\ast}\right) + \sum_{\mathsf{S}_{t+1}}\mathds{P}\left[ \mathsf{S}_{t+1} | \mathsf{S}_{t}, \mathsf{A}_t^{\ast} \right]\widetilde{V}_{t+1}(\mathsf{S}_{t+1}) \nonumber\\
    &\leq&  \mathsf{c}_t\left(\mathsf{S}_t, \mathsf{A}_t^{(t-1)}\right) + \sum_{\mathsf{S}_{t+1}}\mathds{P}\left[ \mathsf{S}_{t+1} | \mathsf{S}_{t}, \mathsf{A}_t^{(t-1)} \right]\widetilde{V}_{t+1}(\mathsf{S}_{t+1})\nonumber
\end{eqnarray}
The optimized scheduling action is applied in the $t$-th time slot if the above inequality holds. Otherwise, the reference scheduling action $\mathsf{A}_t^{(t-1)}=(\boldsymbol{E}_t^{(t-1)},\boldsymbol{\txTime}_t^{(t-1)},\boldsymbol{r}_t^{(t-1)})$ is applied in the $t$-th time slot, that is, $\mathsf{A}_t^{\ast}=\mathsf{A}_t^{(t-1)}$. Finally, the optimized scheduling policies of all time slots are denoted as $\widetilde{\Omega} = \{ \mathsf{A}_t^{\ast} | \forall t \}$. 

\subsection{Update of Reference Scheduling}\label{subsection:update}

Since the actual throughput allocation for the $n$-th vehicle ($\forall n$), $(r_{n,1}^{\ast},r_{n,2}^{\ast},...,r_{n,t}^{\ast})$, might be different from that in the reference scheduling, $(r_{n,1}^{(t-1)},r_{n,2}^{(t-1)},...,r_{n,t}^{(t-1)})$, the throughput allocation in the reference scheduling for the remaining time slots, $(r_{n,t+1}^{(t-1)},...,r_{n,T}^{(t-1)})$, should be updated after the transmission optimization of the $t$-th time slot. 

Particularly, let $r_{n, t + k}^{\mathrm{s}}$ be the auxiliary throughput allocation of the $n$-th vehicle in the $(t+k)$-th time slot as
\begin{equation}
	\label{equation: auxiliary throughput}
	r_{n, t + k}^{\mathrm{s}} \triangleq \begin{cases}
		r_{n, t + k}^{(0)}, & t + k \le t_n^d, \\
		\bar{r}_{n, t + k}(\tau_{n,t+k}^{(0)}), & \text{otherwise},
	\end{cases}
\end{equation}
where $t_n^d$ is the task accomplish time derived in Section \ref{section: offline pre-allocation}.
Then, based on $\Delta_{n,t+1}^k \triangleq d_{n,t+1} - \sum_{\kappa = 1}^{k} r_{n, t + \kappa}^{s}$, the throughput allocation of the $n$-th vehicle in the reference scheduling is updated to
\begin{equation}
	r_{n, t + k}^{(t)} = \begin{cases}
		r_{n, t + k}^{s}, & \Delta_{n,t+1}^k \ge 0, \\
		\left(  \Delta_{n,t+1}^{k-1} \right)^{+}, & \text{otherwise},
	\end{cases},
\end{equation}
where $k\in\{1, \ldots, T-t\}$. As a result, we have the following conclusion on the average total cost of the proposed dynamic improvement solution framework.
\begin{Lemma}[Non-trivial Performance Bound]
The average total cost of the proposed dynamic improvement framework, consisting of the reference scheduling initialization at the beginning of the scheduling period, online optimization of transmission action at the beginning of each time slot and the update of the throughput allocation of the reference scheduling, is upper-bounded by that of the reference scheduling. Thus
\begin{equation}
	\mathsf{C}_t( \mathsf{S}_t, \widetilde{\Omega}  ) \leq \mathsf{C}_t( \mathsf{S}_t, \Omega_{t-1}^H ), \forall t \in \{1, 2, \ldots, T-1\}.
\end{equation}
\end{Lemma}

\begin{proof}
The conclusion is straightforward since the average system cost of the reference scheduling is suppressed by the online optimization of Section \ref{subsection:dynamic-improvement} in each time slot.
\end{proof}
\section{Performance Evaluation}
\label{section: Performance Evaluation}

\begin{figure}[!t]
  \centering
  \includegraphics[width=0.3\textwidth]{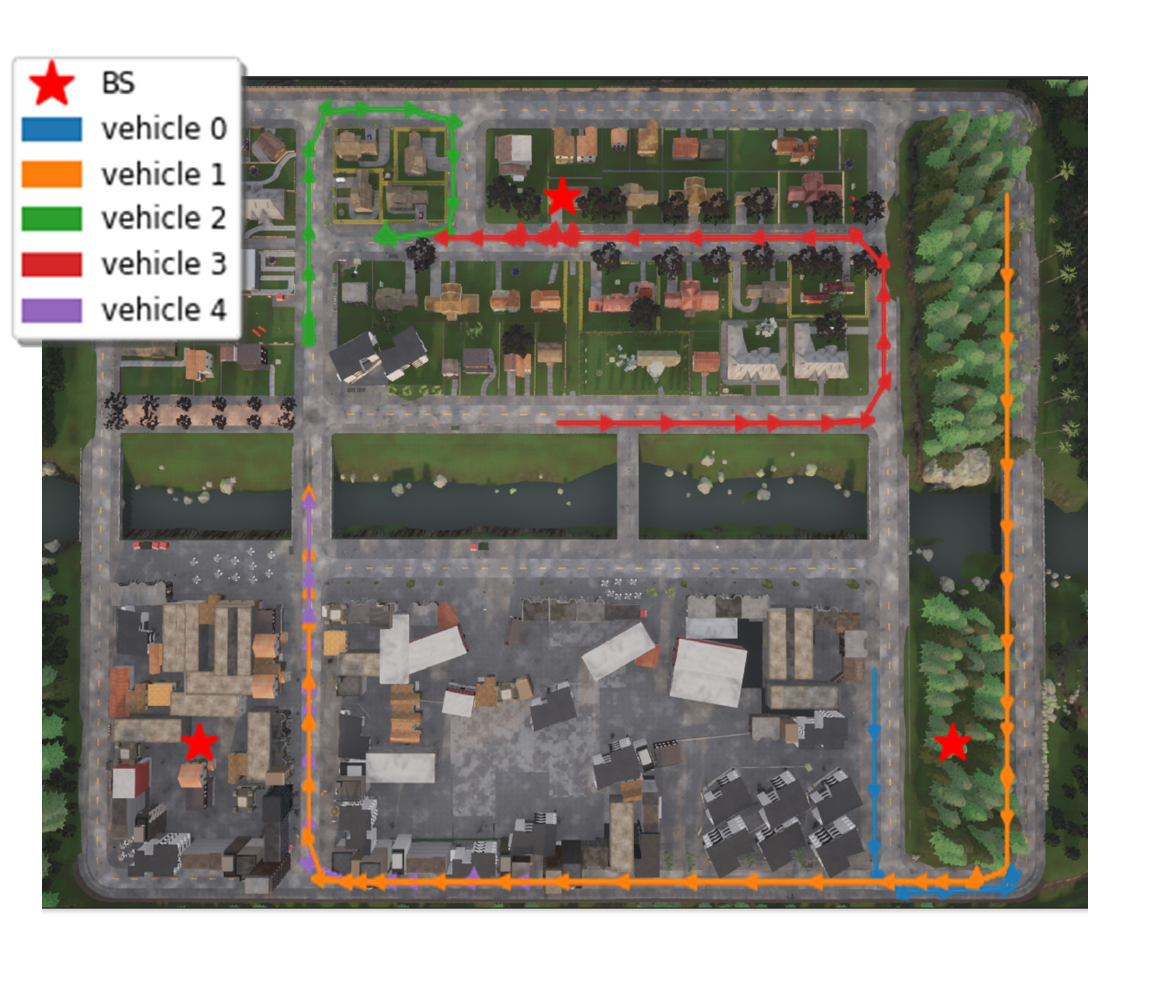}
  \caption{Illustration of simulation scenario, where stars and lines denote BSs' locations and vehicles' trajectories, respectively.}
  \label{figure: illustration of simulator}
  \vspace{-0.5cm}
\end{figure}

\begin{figure}[!t]
  \centering
  \includegraphics[width=0.4\textwidth]{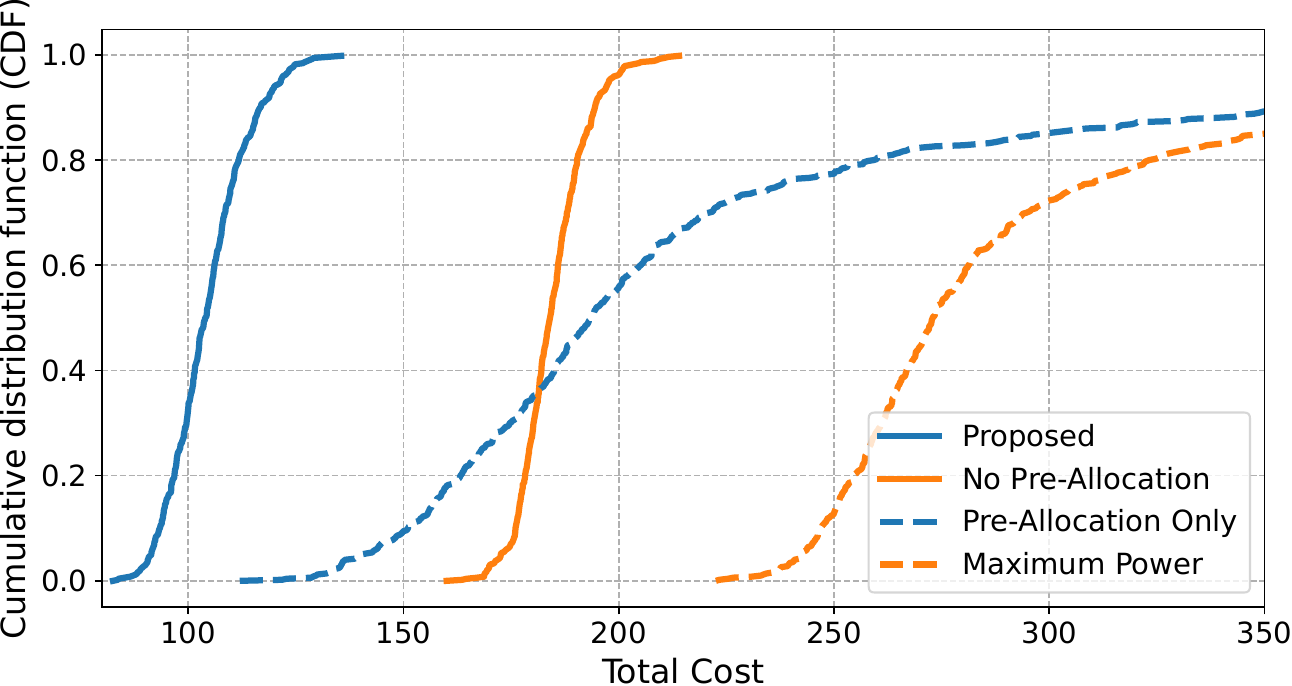}
  \caption{The CDFs of total cost.}
  \label{figure: cdf of cost}
  \vspace{-0.5cm}
\end{figure}

\begin{figure}[!t]
  \centering
  \includegraphics[width=0.4\textwidth]{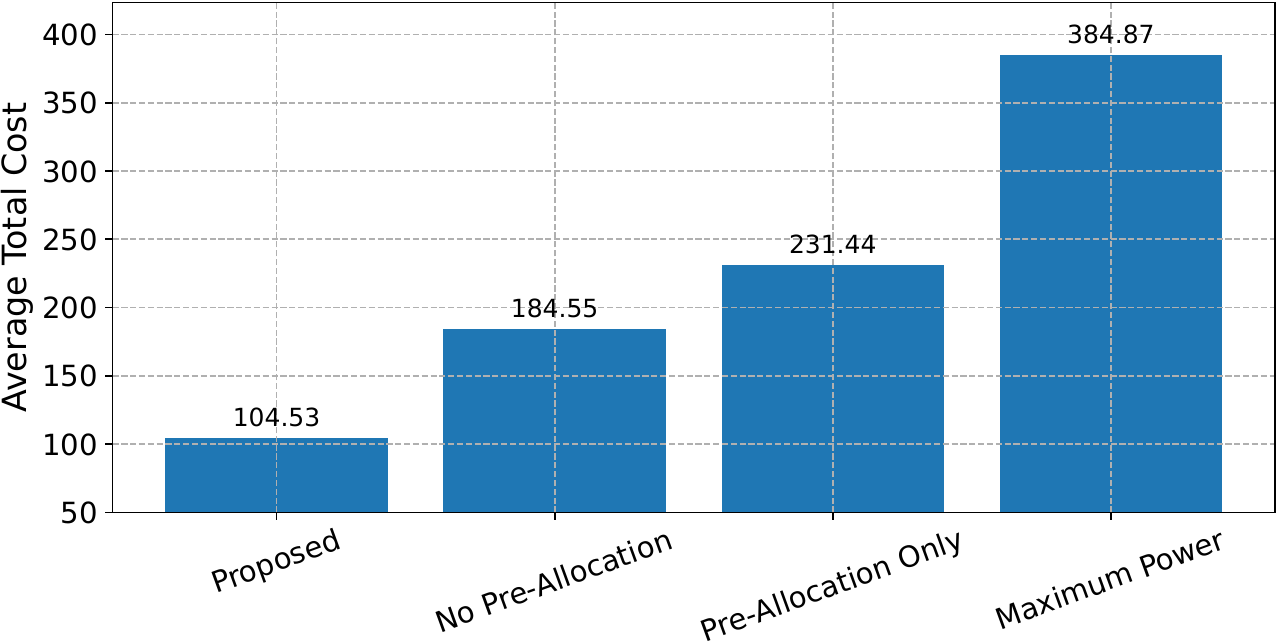}
  \caption{The comparison of average total costs.}
  \label{figure: average cost}
  \vspace{-0.7cm}
\end{figure}

In this section, we evaluate the performance of the proposed solution framework.
The vehicle's trajectories are randomly generated from the CARLA simulator \cite{dosovitskiy2017carla} in the road map \textit{Town01}. 
As shown in \figurename~\ref{figure: illustration of simulator}, the simulation involves $N = 5$ vehicles and $M = 3$ BSs. 
The scheduling period spans $T_p = 50$ seconds, divided into $T = 50$ slots, each lasting $\slotLen = 1$ second.
Each BS operates on a bandwidth of $\bandwidth = 20$ MHz. The pathloss exponent is $\gamma = 4$, and the noise power is $-30$ dBw.
Each vehicle has $360$ Mb information bits arrived at the beginning of the scheduling period, and its peak power constraint is $P_{\max} = 5$ W. 
The weights $\omega_1$ and $\omega_2$ are $2$ and $5$, respectively.

In the simulation, we compare the performance of the proposed framework with the following three baselines, where each algorithm is evaluated $500$ times to obtain the average total cost. 
The three baselines are as follows:
\begin{itemize}
  \item \textbf{Pre-Allocation Only}: The reference policy $\Omega_0^H$ is directly applied in offloading scheduling.
  \item \textbf{Maximum Power}: The $n$-th vehicle is associated with the BS with the best ergodic channel capacity, and transmit with the peak power level. The transmission time of each cell is equally allocated to the associated vehicles.
  \item \textbf{No Pre-Allocation}: Use the above maximum power scheme as the reference scheduling, and optimize the transmission scheduling of each time slot as Section \ref{subsection:dynamic-improvement}. 
\end{itemize}

The performance comparison is made in \figurename~\ref{figure: cdf of cost} and \figurename~\ref{figure: average cost}. Particularly, in \figurename~\ref{figure: cdf of cost}, the cumulative distribution functions (CDFs) of the total costs (obtained via $500$ trials) are illustrated. It can be observed that the proposed framework has significant performance gain over the baselines. Its gain over the scheme of no per-allocation demonstrates that the per-allocation is necessary to provide a good reference scheduling; while its gain over the scheme of pre-allocation only demonstrates the benefits of dynamic improvement. In addition to the CDF, their average total costs are compared in \figurename~\ref{figure: average cost}.
The proposed algorithm outperforms the schemes of pre-allocation only and maximum power with $54.83\%$ and $72.84\%$ average cost suppression, respectively.
Moreover, the proposed algorithm also outperforms the scheme of no pre-allocation with $43.36\%$ average cost suppression.
\section{Conclusion}
\label{sec: Conclusion}
The scheduling of task offloading in vehicular networks is considered in this paper.
The scheduling problem is formulated as an MDP, where the optimal solution suffers from the curse of dimensionality.
To address this issue, a low-complexity solution framework is proposed to achieve a sub-optimal solution with a non-trivial cost upper-bound.
The simulation results based on high-fidelity traffic simulator demonstrate that the proposed framework can achieve a significant performance gain compared with the baselines.

\bibliographystyle{IEEEtran}
\bibliography{IEEEabrv,reference}
\end{document}